\documentclass[12pt]{article}
\usepackage[english]{babel}
\usepackage[utf8]{inputenc}
\usepackage[a4paper,lmargin={3.5cm},rmargin={3.5cm},
tmargin={2.5cm},bmargin = {2.5cm}]{geometry}
\usepackage{amsthm}
\usepackage{graphicx}
\usepackage{hyperref}
\usepackage{mathptmx}
\usepackage{amsmath}
\usepackage{amsfonts}
\usepackage{amssymb}
\usepackage{verbatim}
\usepackage{breqn}
\usepackage{mathrsfs} 
\usepackage{enumitem}
\usepackage{titlesec}

\theoremstyle{definition}
\newtheorem{Def}{Definition}
\theoremstyle{plain}
\newtheorem{Lem}[Def]{Lemma}
\theoremstyle{plain}
\newtheorem{The}[Def]{Theorem}
\theoremstyle{remark}
\newtheorem{Rem}[Def]{Remark}
\theoremstyle{plain}
\newtheorem{Prop}[Def]{Proposition}

\makeatletter
\renewcommand\@maketitle{
\let\footnote\thanks 
\begin{center}
{\normalsize  \@title \par}
\vskip 1em
{\footnotesize \@author \par}
\end{center}
\vspace{-0.3cm}
}
\makeatother

\titleformat*{\section}{\center\normalsize\scshape}
\titlelabel{\thetitle.\hspace{1mm}}
\titlespacing{\section}{0pt}{*4.5}{*1.5}

\title{{ACCUMULATION RATE OF BOUND STATES OF DIPOLES GENERATED BY POINT CHARGES IN STRAINED GRAPHENE}}
\author{FLORIAN DORSCH\footnote{Mathematisches Institut, LMU M{\"u}nchen, Theresienstr. 39, 80333 Munich, Germany\newline \texttt{florian.dorsch@gmx.org}}}

\begin{document}

\maketitle

\begin{abstract}
We consider strained graphene, modelled by the two-dimensional massive Dirac operator, with potentials corresponding to charge distributions with vanishing total charge, non-vanishing dipole moment and finitely many point charges of subcritical coupling constants located in the graphene sheet. We show that the bound state energies accumulate exponentially fast at the edges of the spectral gap by determining the leading order of the accumulation rate.
\end{abstract}

\section{Introduction}

\noindent Electrons close to the Fermi level in strained graphene can be described by the two-dimensional massive Dirac operator \cite{Vozmediano}. De Martino et al \cite{DeMartino} predicted the exis\-tence of infinitely many bound states of the two-dimensional massive Dirac ope\-rator with a dipole potential and that these bound states accumulate with an exponential rate at the edges of the spectral gap. Shortly after, Cuenin and Siedentop~\cite{Cuenin} proved the former statement, whereas the latter one has so far only been proven under the assumption that no point charges are located directly in the graphene sheet (Rademacher and Siedentop \cite{Rademacher}). The purpose of the present article is to extend the result in~\cite{Rademacher} to the case of potentials with finitely many such Coulomb singularities with subcritical coupling constants.\\
The operator of interest acts in $\textnormal{\textsf{L}}^2(\mathbb{R}^2,\mathbb{C}^2)$ and is formally given by the expression
\begin{align}
\begin{split}
F= & \hspace{0.1cm} -i\pmb{\sigma}\cdot\nabla+m\sigma_3+V\,,\\
V= & \hspace{0.1cm} V_{\textnormal{sing}}+V_{\textnormal{reg}}\,,\\
V_{\textnormal{sing}}= & \hspace{0.1cm} \mbox{\footnotesize $\displaystyle\sum\limits_{n=1}^N$}\nu_n|\cdot - x_n|^{-1}
\end{split}
\end{align}
with $\pmb{\sigma}=(\sigma_1,\sigma_2)$, where $\left\{\sigma_k\right\}_{k=1}^3$ are the standard Pauli matrices, $m\in\mathbb{R}^+:=(0,\infty)$ is a strictly positive mass and $V$ is the (real-valued) potential associated to the charge distribution given by a finite, signed Borel measure $\rho$ in $\mathbb{R}^3$ via
\begin{equation}
V:\mathbb{R}^2\rightarrow\mathbb{R}\,, \quad x\mapsto \displaystyle\int\limits_{\mathbb{R}^3}\frac{\textnormal{d}\rho(y)}{|(x,0)-y|}\,,
\end{equation}
where the charge distribution is, accordingly, split into a singular and regular part, viz.,
\begin{equation}
\begin{split}
\rho= & \hspace{0.2cm} \rho_{\textnormal{sing}}+\rho_{\textnormal{reg}}\,,\\
\rho_{\textnormal{sing}}= & \hspace{0.2cm} \mbox{\footnotesize $\displaystyle\sum\limits_{n=1}^N$}\nu_n\delta(\cdot- (x_n,0))\,,
\end{split}
\end{equation}
where the positions $\left\{x_n\right\}_{n=1}^N\subset \mathbb{R}^2$ of the charges of subcritical coupling constants $\left\{\nu_n\right\}_{n=1}^N\subset \left(-1/2,1/2\right)\setminus \{0\}$ are mutually distinct, i.e., $y_k\neq y_j$ whenever $k\neq j$.
The assumptions on $\rho_{\textnormal{reg}}$ will be made in Thm. \ref{main_theorem}.\\

We denote the dipole moment corresponding to $\rho$ by
\begin{equation}\label{dip_mom}
\mathfrak{d}:=\displaystyle\int\limits_{\mathbb{R}^3}(y_1,y_2)\textnormal{d}\rho (y) \in\mathbb{R}^2
\end{equation}
and the radius of a sufficiently large ball around the Coulomb singularities by
\begin{align}
\gamma:=2\max\limits_{n\in\{1,\dots,N\}}\left\{|x_n|\right\}\,.
\end{align}

In order to obtain a physically sensible self-adjoint realization of $F$, we will recall two basic facts proven in \cite{Cuenin}.

\begin{Prop}[distinguished self-adjoint extension (cf. \cite{Cuenin}, Thm. 1, Rem. 1)]\label{dist_sa_ext}
The operator $-i\pmb{\sigma}\cdot\nabla+m\sigma_3+V_{\textnormal{sing}}$ defined on $\textnormal{\textsf{C}}_0^{\infty}\big(\mbox{\small $\mathbb{R}^2\setminus \mbox{\footnotesize $\left\{x_n\right\}_{n=1}^N$},\mathbb{C}^2$}\big)$ has a unique self-adjoint extension $\tilde{D}$ satisfying $\mathscr{D}(\tilde{D})\subset \textnormal{\textsf{H}}^{1/2}(\mathbb{R}^2,\mathbb{C}^2)$.
\end{Prop}

\begin{Prop}[energy gap (\cite{Cuenin}, Prop. 1)]
The essential spectrum of $\tilde{D}$ is given by $\sigma_{\textnormal{ess}}(\tilde{D})=\mathbb{R}\setminus (-m,m)$.
\end{Prop}

To formulate our main result, we will need the following definitions.

\begin{Def}[rescaled Mathieu operator]
For $p\in \mathbb{R}^+$ we define the \textit{rescaled Ma\-thieu operator with periodic boundary conditions} on the domain $\mathscr{D}(M_p)=\textnormal{\textsf{H}}^2(\mathbb{S}^1)$ as
\begin{equation}
M_p\psi:=\big(-\partial^2-p\cos(\cdot)\big)\psi\,.
\end{equation}
\end{Def}

\begin{Def}
For a self-adjoint operator $A$ and a Borel set $I\subset \mathbb{R}\setminus \sigma_{\textnormal{ess}}(A)$ we define the \textit{number of eigenvalues (counting multiplicity)} by 
\begin{equation}
\mathcal{N}_I(A):=\textnormal{rank}(\chi_I(A))\,,
\end{equation}
where $\chi$ denotes the indicator function.\\
\end{Def}

\begin{Def}
We denote a ball of radius $a\in\mathbb{R}_0^+\cup \{\infty\}$ by $B_a:=\left\{x\in\mathbb{R}^2: |x|<a\right\}$.
\end{Def}

\begin{Def}\label{eff_rest_pot}
We introduce the \textit{effective rest potential} $R$, which is obtained from the potential $V$ by subtracting the short range part of $V_{\textnormal{sing}}$ and the long range part of the pure point dipole, i.e.,
\begin{equation}
R:\mathbb{R}^2\rightarrow \mathbb{R}\,, \quad x\mapsto V_{\textnormal{reg}}(x)+\left[V_{\textnormal{sing}}(x)-\frac{\langle\mathfrak{d},x\rangle_{\mathbb{R}^2}}{|x|^3}\right]\chi_{\mathbb{R}^2\setminus B_{\gamma}}(x)\,.
\end{equation}
\end{Def}

The Kato-Rellich theorem implies that
\begin{equation}\label{operator_of_interest}
D:=\tilde{D}+V_{\textnormal{reg}}
\end{equation}
is self-adjoint if the regular part of the potential $V_{\textnormal{reg}}$ is relatively $\tilde{D}$-bounded with relative bound $n_{\tilde{D}}(V_{\textnormal{reg}})<1$. 

\begin{The}[exponential accumulation rate]\label{main_theorem}
Let $\tilde{D}$ be the distinguished self-adjoint extension of $-i\pmb{\sigma}\cdot\nabla+m\sigma_3+V_{\textnormal{sing}}$ defined on $\textnormal{\textsf{C}}_0^{\infty}\big(\mbox{\small $\mathbb{R}^2\setminus \mbox{\footnotesize $\left\{x_n\right\}_{n=1}^N$},\mathbb{C}^2$}\big)$ (see Prop. \ref{dist_sa_ext}).
Assume that the following hypotheses hold:
\begin{enumerate}
\item The regular part of the potential $V_{\textnormal{reg}}$ is relatively $\tilde{D}$-bounded with relative bound $n_{\tilde{D}}(V_{\textnormal{reg}})<1$. \label{hyp1}
\item The square of the regular part of the potential $(V_{\textnormal{reg}})^2$ is relatively compact w.r.t. the Laplacian $-\Delta_{\mathbb{R}^2}$ defined on $\textnormal{\textsf{H}}^2(\mathbb{R}^2)$.\label{hyp2}
\item The dipole moment $\mathfrak{d}$ (see (\ref{dip_mom})) is non-zero: $\mathfrak{d}\neq 0$.\label{hyp3}
\item There are neighborhoods of the positions $\left\{x_n\right\}_{n=1}^N$ of the point charges in which the regular part of the potential $V_{\textnormal{reg}}$ is bounded.\label{hyp4}
\item The effective rest potential $R$ (see Def. \ref{eff_rest_pot}) fulfills the following integrability conditions: \label{hyp5}
\begin{enumerate}
\item $R,R^2\in \textnormal{\textsf{L}}^1(\mathbb{R}^2;\log(2+|x|)\textnormal{d}x)$.
\item $|R|_*,(R^2)_*\in \textnormal{\textsf{L}}^1(\mathbb{R}^+;\log_+(r^{-1})\textnormal{d}r)$.
\end{enumerate}
Here, $(\cdot)_*$ denotes the (non-increasing) spherical rearrangement (see \cite{Shargorodsky}).
\end{enumerate}
Then, $D$, defined by (\ref{operator_of_interest}), satisfies
\begin{align} \label{main_statement}
\lim\limits_{E\nearrow m}\frac{\mathcal{N}_{(-E,E)}(D)}{|\log(m-E)|}=\frac{1}{\pi}\textnormal{tr}\left(\sqrt{\left(M_{2m|\mathfrak{d}|}\right)_-}\hspace{0.5mm}\right)\,,
\end{align}
where $(\cdot)_-:=-\min\{\cdot,0\}$ denotes the negative part.
\end{The}

\begin{Rem}
Hypothesis \ref{hyp5}.(a) in Thm. \ref{main_theorem} ensures that the total charge vanishes.
\end{Rem}

\begin{Rem}
The expression on the right of (\ref{main_statement}) is always strictly positive, since the lowest eigenvalue of $M_p$ is negative for all $p\in\mathbb{R}^+$ (see \cite{McLachlan}, 3.25, diagram).
\end{Rem}

\section{The Two-Dimensional Massless Dirac Operator with Subcritical Coulomb Potential}

\noindent Let $\nu\in (-1/2,1/2)\setminus\{0\}$. Since the differential expression of the two-dimensional massless Dirac operator with Coulomb potential $\mbox{\small $-i\pmb{\sigma}\cdot\nabla +\nu |\cdot|^{-1}$}$ acting in $\mbox{\small $\textnormal{\textsf{L}}^2(B_a,\mathbb{C}^2)$}$, where $a\in\mathbb{R}^+\cup \{\infty\}$, commutes with the total angular momentum $-i\partial_{\varphi}+\frac{1}{2}\sigma_3$, it can be decomposed by a unitary map $\mathcal{U}_a:\textnormal{\textsf{L}}^2(B_a,\mathbb{C}^2) \rightarrow \displaystyle\bigoplus\limits_{\kappa\in\mathbb{Z}+1/2}\textnormal{\textsf{L}}^2((0,a),\mathbb{C}^2)$ into 
\begin{align}
\displaystyle\bigoplus\limits_{\kappa\in\mathbb{Z}+1/2}\mathbf{d}_{\kappa}^{\nu}\,, \qquad \textnormal{where}\quad \mathbf{d}_{\kappa}^{\nu}:=\begin{pmatrix} \frac{\nu}{r} & -\partial_r-\frac{\kappa}{r} \\ \partial_r-\frac{\kappa}{r} & \frac{\nu}{r} \end{pmatrix}
\end{align}
(see \cite{Thaller}, Section 7.3.3).
We define for $\kappa\in\mathbb{Z}+1/2$ the operator 
\begin{align}
d^{\nu}_{\kappa,a}:\textnormal{\textsf{C}}_0^{\infty}((0,a),\mathbb{C}^2) &\rightarrow \textnormal{\textsf{L}}^2((0,a),\mathbb{C}^2)\,,\quad
\psi \mapsto \mathbf{d}^{\nu}_{\kappa}\psi\,.
\end{align}

The fundamental solution of $\mathbf{d}_{\kappa}^{\nu}\Upsilon=0$ in $\mathbb{R}^+$ is a linear combination of $\Upsilon_{\kappa}^{\nu,+}$ and $\Upsilon_{\kappa}^{\nu,-}$, where 
\begin{align}
\Upsilon_{\kappa}^{\nu,\pm}(r)=\begin{pmatrix} -\nu \\ \pm \mbox{\footnotesize $\sqrt{\kappa^2-\nu^2}$}-\kappa \end{pmatrix} r^{\pm\sqrt{\kappa^2-\nu^2}}\,,
\end{align}
and hence $\mathbf{d}_{\kappa}^{\nu}$ is in the limit circle case at $\vartheta\in\mathbb{R}^+$ and in the limit point case at $\infty$ and, moreover, it is in the limit circle case at $0$ if and only if $\mbox{\small $\kappa=\pm 1/2$}$ \cite{Morozov1}. Thus, the deficiency indices of $\mbox{\small $d^{\nu}_{\kappa,\infty}$}$ are $\mbox{\small $(1,1)$}$ in case of $\mbox{\small $\kappa=\pm 1/2$}$ and $\mbox{\small $(0,0)$}$, otherwise~\cite{Morozov1},\cite{Weidmann2}; the deficiency indices of $\mbox{\small $d_{\kappa,\vartheta}^{\nu}$}$ are $\mbox{\small $(2,2)$}$ in case of $\mbox{\small $\kappa=\pm 1/2$}$ and $\mbox{\small $(1,1)$}$, otherwise \cite{Weidmann2}.\\\\
 We define 
\begin{align} D_a^{\nu}: \textnormal{\textsf{C}}_0^{\infty}(B_a\setminus \{0\},\mathbb{C}^2)\rightarrow \textnormal{\textsf{L}}^2(B_a,\mathbb{C}^2)\,, \quad \psi\mapsto \big(-i\pmb{\sigma}\cdot\nabla +\nu |\cdot|^{-1}\big)\psi\,
\end{align}
and denote the distinguished self-adjoint extension (see Prop.~\ref{dist_sa_ext}) of $D_{\infty}^{\nu}$ by $H^{\nu}_{\infty}$. 

\begin{Lem}\label{sa_ext_disc_sp}
For all $\vartheta\in\mathbb{R}^+$ there exists a self-adjoint extension $H^{\nu}_{\vartheta}$ of $D^{\nu}_{\vartheta}$ with discrete spectrum.
\end{Lem}
\begin{proof}
It suffices to find self-adjoint extensions $\left\{\mbox{\small $\hat{d}^{\nu}_{\kappa,\vartheta}$}\right\}_{\kappa\in\mathbb{Z}+1/2}$ of $\left\{\mbox{\small $d^{\nu}_{\kappa,\vartheta}$}\right\}_{\kappa\in\mathbb{Z}+1/2}$ with compact resolvents and the property that 
\begin{align}\label{res_conv_to_zero}
\left\|(\hat{d}^{\nu}_{\kappa,\vartheta}+i\mathbb{I})^{-1}\right\|\longrightarrow 0\quad\textnormal{as}\quad \kappa\rightarrow \pm\infty\,.
\end{align}
In case of $\kappa=\pm 1/2$, the resolvent of any self-adjoint extension of $d^{\nu}_{\kappa,\vartheta}$ is a Hilbert-Schmidt operator, since $\mathbf{d}_{\kappa}^{\nu}$ is in the limit circle case both at $0$ and at $\vartheta$ (see~\cite{Weidmann2}, Prop. 1.6).
Now, let $\kappa\neq \pm 1/2$. A self-adjoint extension of $d^{\nu}_{\kappa,\vartheta}$ is given by
\begin{gather}
\begin{split}
\hat{d}_{\kappa,\vartheta}^{\nu}: \mathscr{D}(\hat{d}_{\kappa,\vartheta}^{\nu}) \rightarrow \textnormal{\textsf{L}}^2((0,\vartheta),\mathbb{C}^2)\,,\quad
\psi \mapsto \mathbf{d}^{\nu}_{\kappa}\psi\,,\qquad\qquad\qquad\\
\textnormal{where}\qquad\quad\mathscr{D}(\hat{d}_{\kappa,\vartheta}^{\nu}):=\big\{\phi\equiv (\phi_1,\phi_2)\in \textnormal{\textsf{L}}^2((0,\vartheta),\mathbb{C}^2)\cap \textnormal{\textsf{AC}}_{\textnormal{loc}}((0,\vartheta),\mathbb{C}^2):\qquad\\
 \mathbf{d}_{\kappa}^{\nu}\phi\in \textnormal{\textsf{L}}^2((0,\vartheta),\mathbb{C}^2),\phi_1(\vartheta)=\phi_2(\vartheta)\big\}\,.\quad
\end{split}
\end{gather}
Indeed, observing that the conditions $\phi_1(\vartheta)=\phi_2(\vartheta)$ and $\left\langle i\sigma_2\phi(\vartheta),\mbox{\small $\Upsilon_{\kappa}^{\nu}$}(\vartheta)\right\rangle_{\mathbb{C}^2}=0$, where
\begin{align}
\Upsilon_{\kappa}^{\nu}:=\big(\mbox{\footnotesize $\sqrt{\kappa^2-\nu^2}$}-\nu+\kappa\big)\Upsilon^{\nu,+}_{\kappa}+\vartheta^{2\sqrt{\kappa^2-\nu^2}}\big(\mbox{\footnotesize $\sqrt{\kappa^2-\nu^2}$}+\nu-\kappa\big)\Upsilon^{\nu,-}_{\kappa}\,,
\end{align}
are equivalent, this follows from Prop. 1.5 in \cite{Weidmann2}, since $\Upsilon_{\kappa}^{\nu}$ solves $\mathbf{d}_{\kappa}^{\nu}\Upsilon=0$.\\
There are two functions $\Omega_{\kappa}^{\nu,\pm}:\mathbb{R}^+\rightarrow\mathbb{C}^2$ with
\begin{align}\label{sol1}
\Omega^{\nu,+}_{\kappa}(r)=r^{\sqrt{\kappa^2-\nu^2}}\begin{pmatrix}\kappa+\mbox{\footnotesize $\sqrt{\kappa^2-\nu^2}$}\\ \nu\end{pmatrix}+\mathcal{O}\left(r^{1+\sqrt{\kappa^2-\nu^2}}\right)
\end{align}
and
\begin{align}\label{sol2}
\Omega^{\nu,-}_{\kappa}(r)=r^{-\sqrt{\kappa^2-\nu^2}}\begin{pmatrix} \nu \\ \kappa+\mbox{\footnotesize $\sqrt{\kappa^2-\nu^2}$}\end{pmatrix} +\mathcal{O}\left(r^{1-\sqrt{\kappa^2-\nu^2}}\right)
\end{align}
\begin{samepage}\noindent
as $r\rightarrow 0$ generating the fundamental solution of $(\mathbf{d}_{\kappa}^{\nu}+i)\Omega=0$ in $\mathbb{R}^+$ (see~\cite{Morozov2}, Lem.~8). For any $c\in (0,\vartheta)$, the restriction of $\Omega^{\nu,-}_{\kappa}$ to $(0,c)$ is not contained in $\textnormal{\textsf{L}}^2((0,c),\mathbb{R}^2)$. Therefore, the integral kernel of $(\hat{d}^{\nu}_{\kappa,\vartheta}+i\mathbb{I})^{-1}$ is given by
\begin{align}
G_{\kappa,\vartheta}^{\nu}:(0,\vartheta)^2\rightarrow \mathbb{C}^{2\times 2}\,,\quad
(x,y) \mapsto \textnormal{const.}\begin{cases}
\Omega_{\kappa}^{\nu,+}(x)\mbox{\footnotesize $\Big($}\Omega_{\kappa,\vartheta}^{\nu}(y)\mbox{\footnotesize $\Big)$}^{\intercal}\,,\quad x<y,\\\\
\Omega_{\kappa,\vartheta}^{\nu}(x)\mbox{\footnotesize $\Big($}\Omega_{\kappa}^{\nu,+}(y)\mbox{\footnotesize $\Big)$}^{\intercal}\,,\quad x\geq y\,,
\end{cases}
\end{align}
where $\Omega_{\kappa,\vartheta}^{\nu}$ is the solution of $(\mathbf{d}_{\kappa}^{\nu}+i)\Omega=0$ satisfying $\big\langle i\sigma_2\mbox{\small $\Omega_{\kappa,\vartheta}^{\nu}$}(\vartheta),\mbox{\small $\Upsilon_{\kappa}^{\nu}$}(\vartheta)\big\rangle_{\mathbb{C}^2}=0$ (see \cite{Weidmann2}, Prop. 1.6). It follows from (\ref{sol1}), (\ref{sol2}) and the continuity of $\Omega_{\kappa}^{\nu,\pm}$ on $(0,\vartheta]$ that the components of $\Omega_{\kappa}^{\nu,+}(r)$ and $\Omega_{\kappa,\vartheta}^{\nu}(r)$ are bounded on $r\in (0,\vartheta)$ by $Cr^{\sqrt{\kappa^2-\nu^2}}$ and $Cr^{-\sqrt{\kappa^2-\nu^2}}$ for some $C\in\mathbb{R}^+$, respectively. Thus, $G_{\kappa,\vartheta}^{\nu}$ lies in $\textnormal{\textsf{L}}^2((0,\vartheta)^2,\mathbb{C}^{2\times 2})$ and hence $\big(\hat{d}^{\nu}_{\kappa,\vartheta}+i\mathbb{I}\big)^{-1}$ is a Hilbert-Schmidt operator.
A core for $\hat{d}^{\nu}_{\kappa,\vartheta}$ is given by $\mathfrak{C}^{\vartheta}:=\left\{\left(\phi_1,\phi_2\right)\in \textnormal{\textsf{C}}_0^{\infty}((0,\vartheta],\mathbb{C}^2): \phi_1(\vartheta)=\phi_2(\vartheta)\right\}$, since the closure of the restriction of $\hat{d}^{\nu}_{\kappa,\vartheta}$ to $\mathfrak{C}^{\vartheta}$ is a  strict extension of $\overline{d^{\nu}_{\kappa,\vartheta}}$. Indeed, for instance, $\mbox{\small $f:=(1,1)\lim\limits_{\mbox{\scriptsize $x\downarrow \max\{\cdot,\mbox{\scriptsize $\vartheta/2$}\} $}}\exp\left[2/(\vartheta-2x)\right]$}\in\mathfrak{C}^{\vartheta}\setminus\mathscr{D}\big(\overline{\mbox{\small $d^{\nu}_{\kappa,\vartheta}$}}\big)$, as $\big\langle i\sigma_2 f(\vartheta),\Upsilon_{\kappa}^{\nu,+}(\vartheta)\big\rangle_{\mathbb{C}^2}\neq 0$ (see \cite{Weidmann2}, Prop. 1.5). Now, let $\psi\equiv\mbox{\small $(\psi_1,\psi_2)$}\in \mathfrak{C}^{\vartheta}$. Then, $2\vartheta\big\|\hat{d}^{\nu}_{\kappa,\vartheta}\psi\big\|\geq |\kappa|\hspace{0.5mm}\|\psi\|$ holds for large $|\kappa|$, hence (\ref{res_conv_to_zero}) is satisfied. Indeed, 
\begin{equation}
\begin{aligned}
\left\|\hat{d}_{\kappa,\vartheta}^{\nu}\psi\right\|^2
&=\overbrace{\left\|\hat{d}_{\mbox{\scriptsize $3/2$},\vartheta}^{\nu}\psi\right\|^2}^{\geq 0}+\displaystyle\int\limits_0^{\vartheta}\textnormal{d}r\hspace{0.5mm} \Bigg(\frac{|2\kappa-3|^2+6(2\kappa-3)}{4r^2}|\psi(r)|^2-\\
& \underbrace{-\frac{2(2\kappa-3)\nu}{r^2}\Re\left[\overline{\psi_1(r)}\psi_2(r)\right]}_{\geq -\frac{|2\kappa-3||\nu|}{r^2}|\psi(r)|^2}-\frac{2\kappa-3}{2r}\partial_r\left[|\psi_1(r)|^2-|\psi_2(r)|^2\right]\Bigg)\\
& \geq \overbrace{\frac{1}{4}\left[|2\kappa-3|^2+6(2\kappa-3)-4|2\kappa-3|\left(|\nu|+\mbox{\small $\frac{1}{2}$}\right)\right]\displaystyle\int\limits_0^{\vartheta}\textnormal{d}r\hspace{0.5mm}\left|\frac{\psi(r)}{r}\right|^2}^{=:S^{\nu}_{\kappa}[\psi]}+\\
& + \frac{1}{2}\underbrace{\left(|2\kappa-3|\displaystyle\int\limits_0^{\vartheta}\textnormal{d}r\hspace{0.5mm} \left|\frac{\psi(r)}{r}\right|^2-(2\kappa-3)\displaystyle\int\limits_0^{\vartheta}\textnormal{d}r\hspace{0.5mm} \frac{1}{r}\partial_r\left[|\psi_1(r)|^2-|\psi_2(r)|^2\right]\right)}_{=:T_{\kappa}[\psi]}\,.
\end{aligned}
\end{equation}
Exploiting the boundary condition at $r=\vartheta$, integration by parts yields
\begin{equation}
\displaystyle\int\limits_0^{\vartheta}\textnormal{d}r\hspace{0.5mm} \frac{1}{r}\partial_r\left[|\psi_1(r)|^2-|\psi_2(r)|^2\right]=\displaystyle\int\limits_0^{\vartheta}\textnormal{d}r\hspace{0.5mm}\frac{1}{r^2}\left[|\psi_1(r)|^2-|\psi_2(r)|^2\right]
\end{equation}
and, therefore, $T_{\kappa}[\psi]$ is nonnegative. But for large $|\kappa|$ it holds that
\begin{equation}
4S_{\kappa}^{\nu}[\psi]\geq \kappa^2\displaystyle\int\limits_0^{\vartheta}\textnormal{d}r\hspace{0.5mm} \left|\frac{\psi(r)}{r}\right|^2\geq \frac{\kappa^2}{\vartheta^2}\|\psi\|^2\,.
\end{equation}
\end{samepage}
\end{proof}

\begin{Lem}\label{Coulomb_Dirac_core}
For all $\vartheta\in\mathbb{R}^+$ there exist $f_{\vartheta},g_{\vartheta}\in \textnormal{\textsf{H}}^{1/2}(B_{\vartheta},\mathbb{C}^2)$ so that the restriction of $H^{\nu}_{\infty}$ to $\textnormal{\textsf{C}}_0^{\infty}(\mathbb{R}^2\setminus \{0\},\mathbb{C}^2)\dotplus \textnormal{span}\left\{f_{\vartheta},g_{\vartheta}\right\}$ is essentially self-adjoint.
\begin{proof}
For $\kappa\neq \pm 1/2$ the operator $d^{\nu}_{\kappa,\infty}$ is essentially self-adjoint (see above).
Now, let $\kappa=\pm 1/2$. We pick $\varsigma\in \textnormal{\textsf{C}}_0^{\infty}([0,\vartheta))$ such that $\varsigma(2w)=1 \quad \forall w\leq \vartheta$. Then, it holds that $\varsigma\Upsilon^{\nu,\pm}_{\kappa}\not\in\mathscr{D}\big(\mbox{\small $\overline{d^{\nu}_{\kappa,\infty}}$}\big)$, since $\lim\limits_{r\rightarrow 0}\big\langle i\sigma_2\varsigma(r)\mbox{\small $\Upsilon^{\nu,\pm}_{\kappa}$}(r),\mbox{\small $\Upsilon^{\nu,\mp}_{\kappa}$}(r)\big\rangle_{\mathbb{C}^2}\neq 0$ (see \cite{Weidmann2}, Prop. 1.5), and therefore
\begin{align}
\mathscr{D}\mbox{\footnotesize $\Big($}\left(d^{\nu}_{\kappa,\infty}\right)^*\mbox{\footnotesize $\Big)$}=\mathscr{D}\mbox{\footnotesize $\Big($}\mbox{\small $\overline{d^{\nu}_{\kappa,\infty}}$}\mbox{\footnotesize $\Big)$}\dotplus \textnormal{span}\Big\{\varsigma\Upsilon^{\nu,+}_{\kappa},\varsigma\Upsilon^{\nu,-}_{\kappa}\Big\}\,.
\end{align}
Thus, any self-adjoint extension of $d^{\nu}_{\kappa,\infty}$ is obtained as the closure of the restriction of $\left(d^{\nu}_{\kappa,\infty}\right)^*$ to $\textnormal{\textsf{C}}_0^{\infty}(\mathbb{R}^+,\mathbb{C}^2)\dotplus\textnormal{span}\big\{\varsigma\big(\alpha\mbox{\small $\Upsilon^{\nu,+}_{\kappa}$}+\beta\mbox{\small $\Upsilon^{\nu,-}_{\kappa}$}\big)\big\}$ with $(\alpha,\beta)\in\mathbb{C}^2\setminus\{0\}$; the distinguished one is obtained by setting $(\alpha,\beta)=(1,0)$. Indeed, one may \mbox{verify that} $\mathcal{U}_{\infty}^*\mathfrak{P}_{\kappa}\big[\varsigma\big(\alpha\mbox{\small $\Upsilon^{\nu,+}_{\kappa}$}+\beta\mbox{\small $\Upsilon^{\nu,-}_{\kappa}$}\big)\big]\not\in \textnormal{\textsf{L}}^{2}(\mathbb{R}^2,\mathbb{C}^2;|x|^{-1}\textnormal{d}x)$ for all $(\alpha,\beta)\in\mathbb{C}^2\setminus\textnormal{span}\left\{(1,0)\right\}$, where $\mathfrak{P}_{\kappa}:\textnormal{\textsf{L}}^2(\mathbb{R}^+,\mathbb{C}^2)\rightarrow \mbox{\small $\bigoplus\limits_{\mbox{\tiny $\lambda\in\mathbb{Z}+1/2$}}$}\textnormal{\textsf{L}}^2(\mathbb{R}^+,\mathbb{C}^2)$, $f\mapsto\mbox{\small $\bigoplus\limits_{\mbox{\tiny $\lambda\in\mathbb{Z}+1/2$}}$} \delta_{\lambda,\kappa}f$, where $\delta_{\cdot,\cdot}$ is the Kronecker symbol. But $\textnormal{\textsf{H}}^{1/2}(\mathbb{R}^2)\subset \textnormal{\textsf{L}}^2(\mathbb{R}^2;(1+|x|^{-1})\textnormal{d}x)$ (cf. \cite{Mueller}, Rem. 15).
\end{proof}
\end{Lem}

\section{Proof of the Exponential Accumulation Rate}

\noindent It follows from Hypothesis \ref{hyp4} in Thm. \ref{main_theorem} that we can choose
\begin{align}
\epsilon\in\left(0,\min\left\{|x_j-x_k|/3:j,k\in\{1,\dots,N\},j\neq k\right\}\right)
\end{align}
such that $V_{\textnormal{reg}}$ is bounded on $\left(B_{\epsilon}+x_n\right)$ for all $n\in\{1,\dots,N\}$.\\
In order to localize the Coulomb singularities, we pick a partition of unity $\mbox{\small $\left\{U_n\right\}_{n=0}^N$}$ with the following properties:
\begin{itemize}
\item $\left\{U_n\right\}_{n=0}^N\subset \textnormal{\textsf{C}}^{\infty}(\mathbb{R}^2,[0,1])\,,$
\item $\mbox{\footnotesize $\displaystyle\sum\limits_{n=0}^N$}(U_n)^2=1\,,$
\item $\textnormal{supp}(U_n)\in B_{\epsilon}+x_n$ for all $n\in\{1,\dots,N\}\,,$
\item $\textnormal{supp}(U_0)\subset \mathbb{R}^2\setminus \mbox{\footnotesize $\displaystyle\bigcup\limits_{n=1}^N$}\left(\overline{B_{\epsilon/2}}+x_n\right)\,.$
\end{itemize}

Thanks to the spectral theorem, we can devote ourselves to the study of the asymptotic behavior of the negative eigenvalues of $D^2-m^2\mathbb{I}\,$ using the Min-Max principle (see~\cite{Schmuedgen}, Thm. 12.1). The Min-Max values of a lower semi-bounded self-adjoint operator $A$ are denoted by $\mu_{(\cdot)}(A)$ (see~\cite{Schmuedgen}, formula (12.2)). The following lemma guarantees that we can restrict the minimization to $\mathfrak{C}:=\textnormal{\textsf{C}}_0^{\infty}\big(\mbox{\small $\mathbb{R}^2\setminus \mbox{\footnotesize $\left\{x_n\right\}_{n=1}^N$},\mathbb{C}^2$}\big)$ when investigating the asymptotic behavior of the eigenvalues.

\begin{Lem}\label{finite_defect}
The defect number of the restriction of $D$ to $\mathfrak{C}$ is bounded above by $2N$.
\end{Lem}
\begin{proof}
It follows from Lem.~\ref{Coulomb_Dirac_core} that for all $n\in\{0,\dots,N\}$ the restriction of $U_n \tilde{D}U_n$ to $\mbox{$\tilde{\mathfrak{C}}:=\mbox{ $\mathfrak{C}\dotplus \textnormal{span}$}\mbox{\footnotesize $\Big($}\mbox{$\big\{f_{\mbox{\scriptsize $\epsilon$}}(\cdot-x_j),g_{\mbox{\scriptsize $\epsilon$}}(\cdot-x_j)\big\}_{\mbox{\scriptsize $j=1$}}^{\mbox{\scriptsize $N$}}$}\mbox{\footnotesize $\Big)$}$}$ is essentially self-adjoint. Hence, given $\psi\in\mathscr{D}(\tilde{D})$ and $n\in\{0,\dots,N\}$, we can choose a sequence $\left\{\mbox{\footnotesize $\psi^{(n)}_k$}\right\}_{k\in\mathbb{N}}\subset\tilde{\mathfrak{C}}$ such that $\psi^{(n)}_k\rightarrow\psi$ as $k\rightarrow\infty$ w.r.t. $\|\cdot\|_{ U_n\tilde{D} U_n}$. For $k\in\mathbb{N}$ we define $\phi_k:=\mbox{\small $\sum\limits_{n=0}^N (U_n)^2\psi^{(n)}_k$}$ and observe that $\left\{\phi_k\right\}_{k\in\mathbb{N}}\subset\tilde{\mathfrak{C}}$ and $\phi_k\rightarrow\psi$ as $k\rightarrow\infty$ w.r.t. $\|\cdot\|_{\tilde{D}}$. Thus, $\tilde{\mathfrak{C}}$ is a core for $\tilde{D}$ and - by the Kato-Rellich theorem~\cite{Reed2} and Hypothesis \ref{hyp1} in Thm.~\ref{main_theorem} - also \mbox{for $D$}.
\end{proof}

Following Rademacher and Siedentop~\cite{Rademacher}, we will deal with Schr{\"o}dinger opera\-tors with relatively compact perturbations of $-\Delta_{\mathbb{R}^2}:\textnormal{\textsf{H}}^2(\mathbb{R}^2)\rightarrow \textnormal{\textsf{L}}^2(\mathbb{R}^2)\,,\psi\mapsto -\Delta\psi$
whose eigenvalues bound those of $D^2-m^2\mathbb{I}$. It follows by the inequality of Seiler and Simon (see~\cite{Seiler}, Lem.~2.1) that $W_2\left(\mbox{\small $-\Delta_{\mathbb{R}^2}$}+\mathbb{I}\right)^{-1}$ is a Hilbert-Schmidt operator for all $W_2\in\textnormal{\textsf{L}}^2(\mathbb{R}^2)$. As the operator norm limit preserves compactness, any potential that lies in
\begin{align}
\begin{split}
\textnormal{\textsf{L}}^{2}_{\infty}(\mathbb{R}^2):=\Big\{W:\mathbb{R}^2\rightarrow\mathbb{C}: \hspace{0.5mm}\forall \epsilon\in\mathbb{R}^+ \hspace{1.5mm}\exists (W_2,W_{\infty})\in \textnormal{\textsf{L}}^2(\mathbb{R}^2)\times\textnormal{\textsf{L}}^{\infty}(\mathbb{R}^2)\\
 \textnormal{such that } W=W_2+W_{\infty} \textnormal{ and } \|W_{\infty}\|_{\infty}<\epsilon \Big\}
\end{split}
\end{align}
is relatively compact w.r.t. $-\Delta_{\mathbb{R}^2}$. For any such $(\mbox{\small $-\Delta_{\mathbb{R}^2}$})$-compact potential $W$, the operator $-\Delta_{\mathbb{R}^2}+W$ is bounded from below and its restriction to $\textnormal{\textsf{C}}_0^{\infty}(\mathbb{R}^2)$ is essentially self-adjoint (see \cite{Weidmann1}, Section 17.2).\\
Hypothesis \ref{hyp5}.(a) in Thm. \ref{main_theorem} implies that $V_{\textnormal{reg}}\in \textnormal{\textsf{L}}^2_{\infty}(\mathbb{R}^2)$. Therefore, $V\chi_{\textnormal{supp}(U_0)}$ and $W:=V_{\textnormal{sing}}\left[V_{\textnormal{sing}}+2V_{\textnormal{reg}}\right]\chi_{\textnormal{supp}(U_0)}$ lie in $\textnormal{\textsf{L}}^2_{\infty}(\mathbb{R}^2)$, since  $V_{\textnormal{sing}}\chi_{\textnormal{supp}(U_0)}$ is bounded and vanishes at infinity. Furthermore, $(V_{\textnormal{reg}})^2$ is $(\mbox{\small $-\Delta_{\mathbb{R}^2}$})$-compact (see Hypothesis \ref{hyp2} in Thm.~\ref{main_theorem}). Thus, both $V\chi_{\textnormal{supp}(U_0)}$ and $V^2\chi_{\textnormal{supp}(U_0)}=(V_{\textnormal{reg}})^2\chi_{\textnormal{supp}(U_0)}+W$ are relatively compact perturbations of $-\Delta_{\mathbb{R}^2}$ and the operator
\begin{align}
\mathfrak{V}^{\zeta}_{\pm}:=-\Delta_{\mathbb{R}^2}+\Big(\pm 2mV\chi_{\textnormal{supp}(U_0)}+(1-1/\zeta)V^2\chi_{\textnormal{supp}(U_0)}-\mbox{\footnotesize $\displaystyle\sum\limits_{n=0}^N$}|\nabla U_n|^2\Big)/(1-\zeta)
\end{align}
is self-adjoint and bounded from below for all $\zeta\in (0,1)$. Moreover, $\textnormal{\textsf{C}}_0^{\infty}(\mathbb{R}^2)$ is a core for $\mathfrak{V}^{\zeta}_{\pm}$. Obviously, $V\chi_{\mathbb{R}^2\setminus B_{\gamma}}\in\textnormal{\textsf{L}}^2_{\infty}(\mathbb{R}^2)$ and, therefore, the operator 
\begin{align}
\mbox{\small $\tilde{\mathfrak{W}}^{\zeta}_{\pm}:\textnormal{\textsf{C}}^{\infty}_0(\mathbb{R}^2\setminus\overline{B_{\gamma}})\rightarrow \textnormal{\textsf{L}}^2(\mathbb{R}^2\setminus B_{\gamma})\,,\quad  \psi\mapsto \left[-\Delta+\left[\pm 2mV+(1+1/\zeta)V^2\right]/(1+\zeta)\right]\psi $}
\end{align}
is bounded from below for all $\zeta\in (0,1)$. We denote its Friedrichs extension (see~\cite{Reed2}, Thm. X.23) by $\mathfrak{W}^{\zeta}_{\pm}$. A form core for $\mathfrak{W}^{\zeta}_{\pm}$ is given by $\textnormal{\textsf{C}}^{\infty}_0(\mathbb{R}^2\setminus\overline{B_{\gamma}})$.

In the following lemma, we reduce the problem to the study of negative eigenvalues of $\mathfrak{W}^{\zeta}_{\pm}$ and $\mathfrak{V}^{\zeta}_{\pm}$.

\begin{Lem}\label{lem_bounds_number_of_ev}
There exist $c\in\mathbb{N}$ such that for all $\zeta\in (0,1)$ and $\upsilon\in\mathbb{R}^-:=(-\infty,0)$ it holds that
\begin{align}\label{bounds_for_number_of_ev}
\mbox{\small $\displaystyle\sum\limits_{+,-}$} \mathcal{N}_{(-\infty,\upsilon)}\mbox{\footnotesize $\Big(\mathfrak{W}^{\zeta}_{\pm}\Big)$}\leq\mathcal{N}_{(-\infty,\upsilon)}(D^2-m^2\mathbb{I})\leq c+\displaystyle\sum\limits_{+,-}\mathcal{N}_{(-\infty,\upsilon)}\mbox{\footnotesize $\Big(\mathfrak{V}^{\zeta}_{\pm}\Big)$}\,.
\end{align}
\end{Lem}
\begin{proof}
We first claim the existence of $\{\tilde{c}_n\}_{n=1}^N\in\big(\mathbb{R}^+\big)^N$ such that the inequality
\begin{align}\label{preliminary_claim}
\mbox{\small $\mu_{\mbox{\footnotesize $\displaystyle\sum\limits_{n=0}^Ns_n+N$}}\mbox{\footnotesize $\Big($}D^2-m^2\mathbb{I}\mbox{\footnotesize $\Big)$}\geq (1-\zeta)\min\left\{0,{\mu}_{s_0}\left(\mathfrak{V}^{\zeta}_+\oplus \mathfrak{V}^{\zeta}_-\right)\right\}+\frac{1}{2}\displaystyle\sum\limits_{n=1}^N\min\left\{0,\mu_{s_n}\mbox{\footnotesize $\Big($}\big(\mbox{\footnotesize $H_{\epsilon}^{\nu_n}$}\big)^2\mbox{\footnotesize $\Big)$}-\tilde{c}_n\right\}$}
\end{align}
holds for all $\zeta\in (0,1)$ and $\{s_n\}_{n=0}^N\in \mathbb{N}^{N+1}$.

One easily checks that the IMS localization formula for the quadratic form associated to $D^2-m^2\mathbb{I}$,
\begin{align}
\begin{split}
\|D\cdot\|^2-m^2\|\cdot\|^2 & =\mbox{\small $\displaystyle\sum\limits_{n=0}^N$}\mathfrak{v}[ U_n\hspace{0.1cm}\cdot\hspace{0.1cm}]\,,\\
\mathfrak{v}:\mathscr{D}(D)\rightarrow\mathbb{R}\,,&\quad \psi\mapsto\|D\psi\|^2-m^2\|\psi\|^2-\mbox{\small $\displaystyle\sum\limits_{j=0}^N$}\Big\||\nabla U_j|\psi\Big\|^2\,,
\end{split}
\end{align}
holds, and therefore Lem. \ref{finite_defect} implies that 
\begin{align}\label{loc_est_1}
\mu_{\mbox{\footnotesize $\displaystyle\sum\limits_{n=0}^Ns_n+N$}}\left(D^2-m^2\mathbb{I}\right)\geq \sup\limits_{M\subset \textnormal{\scriptsize \textsf{L}\normalsize}^2(\mathbb{R}^2,\mathbb{C}^2)\atop \textnormal{dim}(\textnormal{span}(M))\leq \mbox{\scriptsize $\displaystyle\sum\limits_{n=0}^Ns_n-N-1$}}\inf\limits_{\psi\in \mathfrak{C}\cap M^{\perp}\atop \|\psi\|=1}\displaystyle\sum\limits_{j=0}^N\mathfrak{v}[U_j\psi]\,.
\end{align}
The estimate 
\begin{align}
\mbox{\small $\sup\limits_{M\subset \textnormal{\tiny \textsf{L}\normalsize}^2(\mathbb{R}^2,\mathbb{C}^2)\atop \textnormal{dim}(\textnormal{span}(M))\leq \mbox{\scriptsize $\displaystyle\sum\limits_{n=0}^Ns_n-N-1$}}\inf\limits_{\psi\in \mathfrak{C}\cap M^{\perp}\atop \|\psi\|=1}\displaystyle\sum\limits_{j=0}^N\mathfrak{v}[ U_j\psi]  \geq \displaystyle\sum\limits_{n=0}^N\sup\limits_{M_n\subset \textnormal{\tiny \textsf{L}\normalsize}^2(\mathbb{R}^2,\mathbb{C}^2)\atop \textnormal{dim}(\textnormal{span}(M_n))\leq s_n-1}\inf\limits_{\psi\in \mathfrak{C}\cap M^{\perp}_n\atop \|\psi\|=1}\mathfrak{v}[ U_n\psi]$}
\end{align}
is trivial. Partially following Evans et al~\cite{Evans} (inequality (21)), we obtain 
\begin{align}\label{loc_est_2}
\sup\limits_{M_n\subset \textnormal{\scriptsize \textsf{L}\normalsize}^2(\mathbb{R}^2,\mathbb{C}^2)\atop \textnormal{dim}(\textnormal{span}(M_n))\leq s_n-1}\inf\limits_{\psi\in \mathfrak{C}\cap M_n^{\perp}\atop \|\psi\|=1}\mathfrak{v}[ U_n\psi]\geq \sup\limits_{M_n\subset \textnormal{\scriptsize \textsf{L}\normalsize}^2(\mathbb{R}^2,\mathbb{C}^2)\atop \textnormal{dim}(\textnormal{span}(M_n))\leq s_n-1}\inf\limits_{\psi\in \mathfrak{C}\cap (U_n\mathscr{A}_nM_n)^{\perp}\atop \|\psi\|=1}\mathfrak{v}[ U_n\psi]\,,
\end{align}
where $\mathscr{A}_n:\textnormal{\textsf{L}}^2(\mathbb{R}^2,\mathbb{C}^2)\rightarrow \textnormal{\textsf{L}}^2(\mathbb{R}^2,\mathbb{C}^2)\,,\textnormal{ }\psi\mapsto \psi(\cdot -x_n)\,$, where we set $x_0:=0$, and
\begin{align}\label{loc_est_3}
\mbox{\small $\inf\limits_{\psi\in\mathfrak{C}\cap(U_n\mathscr{A}_nM_n)^{\perp}\atop \|\psi\|=1}\mathfrak{v}[ U_n\psi]\geq\inf\limits_{\psi\in U_n(\mathfrak{C}\cap  (U_n\mathscr{A}_nM_n)^{\perp})\atop \|\psi\|\leq 1}\mathfrak{v}[\psi]
\geq \inf\limits_{\psi\in U_n\mathfrak{C}\cap \mathscr{A}_nM^{\perp}_n\atop \|\psi\|\leq 1}\mathfrak{v}[\psi]=\min\Bigg\{0,\inf\limits_{\psi\in U_n\mathfrak{C}\cap \mathscr{A}_nM^{\perp}_n\atop \|\psi\|= 1}\mathfrak{v}[\psi]\Bigg\}\,.$}
\end{align}
The second step in (\ref{loc_est_3}) follows from the inclusion $U_n\left(U_n\mathscr{A}_nM_n\right)^{\perp}\subset \mathscr{A}_nM^{\perp}_n$.\\
As $V\chi_{\textnormal{supp}(U_0)}\in \textnormal{\textsf{L}}^2_{\infty}(\mathbb{R}^2)$ (see Hypothesis \ref{hyp5}.(a) in Thm. \ref{main_theorem}), all $\psi\in \textnormal{\textsf{C}}_0^{\infty}(\mathbb{R}^2,\mathbb{C}^2)$ obey
\begin{align}\label{cs_inequality}
\begin{split}
|2\Re\left\langle -i\pmb{\sigma}\cdot \nabla\psi, V\chi_{\textnormal{supp}(U_0)}\psi\right\rangle | & \leq 2 \|\nabla\psi\| \|V\chi_{\textnormal{supp}(U_0)}\psi\|\leq \\
&\leq \left\langle\psi,\left(-\zeta\Delta+V^2\chi_{\textnormal{supp}(U_0)}/\zeta\right)\psi\right\rangle
\end{split}
\end{align}
(cf. \cite{Rademacher}, inequality (14)). If $\psi\equiv (\psi_1,\psi_2)\in U_0\mathfrak{C}\subset \textnormal{\textsf{C}}_0^{\infty}(\textnormal{supp}(U_0),\mathbb{C}^2)$, then, with $\psi=\chi_{\textnormal{supp}(U_0)}\psi$, (\ref{cs_inequality}) implies that
\begin{align}\label{cs_ineq_conseq}
\begin{split}
\mbox{\small $\|D\psi\|^2-m^2\|\psi\|^2$} & \geq \left\langle\psi,\left((\zeta-1)\Delta + 2mV\chi_{\textnormal{supp}(U_0)}\sigma_3+(1-1/\zeta)V^2\chi_{\textnormal{supp}(U_0)}\right)\psi\right\rangle\\
& = \mbox{\footnotesize $(1-\zeta)\left\langle \psi,\left(-\Delta+\left[2mV\chi_{\textnormal{supp}(U_0)}\sigma_3+(1-1/\zeta)V^2\chi_{\textnormal{supp}(U_0)}\right]/(1-\zeta)\right)\psi\right\rangle $}\\
& =\mbox{\small $(1-\zeta) \Big\langle \psi_1\oplus\psi_2,\left(\mathfrak{V}^{\zeta}_+\oplus\mathfrak{V}^{\zeta}_-\right)\psi_1\oplus\psi_2\Big\rangle + \displaystyle\sum_{n=0}^N\Big\||\nabla U_n|\psi\Big\|^2$}
\end{split}
\end{align}
(cf. \cite{Rademacher}, inequality (16)), which is equivalent to
\begin{align}\label{cs_ineq_concl}
\mathfrak{v}[\psi]\geq (1-\zeta)\Big\langle \psi_1\oplus\psi_2,\left(\mathfrak{V}^{\zeta}_+\oplus\mathfrak{V}^{\zeta}_-\right)\psi_1\oplus\psi_2\Big\rangle \,.
\end{align}
It follows from (\ref{loc_est_3}), (\ref{cs_ineq_concl}) and $U_0\mathfrak{C}\subset \textnormal{\textsf{C}}^{\infty}_0(\mathbb{R}^2,\mathbb{C}^2)$ that 
\begin{align}\label{bound_regular_part}
\mbox{\small $\inf\limits_{\psi\in\mathfrak{C}\cap (U_0M_0)^{\perp}\atop \|\psi\|=1}\mathfrak{v}[U_0\psi]\geq (1-\zeta)\min\Bigg\{0,\inf\limits_{\psi_1\oplus\psi_2\in \textnormal{\tiny \textsf{C}\normalsize}^{\infty}_0(\mathbb{R}^2,\mathbb{C}^2)\cap M_0^{\perp}\atop \|\psi_1\oplus\psi_2\|=1}\left\langle \psi_1\oplus\psi_2,\left(\mathfrak{V}^{\zeta}_+\oplus\mathfrak{V}^{\zeta}_-\right)\psi_1\oplus\psi_2\right\rangle\Bigg\}$}\,.
\end{align}

Now, let $n\in\{1,\dots,N\}$. Hypothesis \ref{hyp4} in Thm. \ref{main_theorem} guarantees the existence of $c^{\prime}_n\in\mathbb{R}^+$ satisfying
\begin{align}
\mbox{\small $\|D\psi\|^2\geq \left(\left\|\left(-i\pmb{\sigma}\cdot\nabla+\nu_n|\cdot-x_n|^{-1}\right)\psi\right\|^2-c^{\prime}_n\|\psi\|^2\right)/2=\left(\|D_{\epsilon}^{\nu_n}\mathscr{A}_n^*\psi\|^2-c^{\prime}_n\|\psi\|^2\right)/2$}
\end{align}
for all $\psi\in U_n\mathfrak{C}\subset\textnormal{\textsf{C}}^{\infty}_0((B_{\epsilon}\setminus\{0\})+x_n)=\mathscr{A}_n\mathscr{D}\big(D^{\nu_n}_{\epsilon}\big)$. As $D^{\nu_n}_{\epsilon}\subset H^{\nu_n}_{\epsilon}$, it follows that\vspace{3mm}
\begin{align}\label{est_by_extension}
\mbox{\footnotesize $\inf\limits_{\psi\in  U_n\mathfrak{C}\cap \mathscr{A}_nM^{\perp}_n\atop \|\psi\|=1}\mathfrak{v}[\psi]\geq \frac{1}{2}\inf\limits_{\psi\in  \mathscr{A}_n\mbox{\footnotesize $\big($}\mathscr{D}\mbox{\scriptsize $\big($}D^{\nu_n}_{\epsilon}\mbox{\scriptsize $\big)$}\cap M^{\perp}_n\mbox{\footnotesize $\big)$}\atop \|\psi\|=1}\Big(\|D_{\epsilon}^{\nu_n}\mathscr{A}^*_n\psi\|^2-\tilde{c}_n\Big)\geq \frac{1}{2}\inf\limits_{\psi\in \mathscr{D}\mbox{\scriptsize $\big($}\mbox{\tiny $H_{\epsilon}^{\nu_n}$}\mbox{\scriptsize $\big)$}\cap M^{\perp}_n\atop \|\psi\|=1}\Big(\|H_{\epsilon}^{\nu_n}\psi\|^2-\tilde{c}_n\Big)$}
\end{align}
holds for some $\tilde{c}_n\in\mathbb{R}^+$. Plugging (\ref{est_by_extension}) into (\ref{loc_est_3}), we obtain
\begin{align}\label{bound_singular_part}
\inf\limits_{\psi\in\mathfrak{C}\cap (U_n\mathscr{A}_nM_n)^{\perp}\atop \|\psi\|=1}\mathfrak{v}[U_n\psi]\geq  \frac{1}{2}\min\Bigg\{0,\inf\limits_{\psi\in \mathscr{D}\mbox{\footnotesize $\big($}\mbox{\scriptsize $H_{\epsilon}^{\nu_n}$}\mbox{\footnotesize $\big)$}\cap M^{\perp}_n\atop \|\psi\|=1}\Big(\|H_{\epsilon}^{\nu_n}\psi\|^2-\tilde{c}_n\Big)\Bigg\}\,.
\end{align}
Then, our preliminary claim (inequality (\ref{preliminary_claim})) follows from (\ref{loc_est_1})-(\ref{loc_est_2}), (\ref{bound_regular_part}) and (\ref{bound_singular_part}).
With $s_0=\mathcal{N}_{(-\infty,\upsilon/(1-\zeta))}\mbox{\footnotesize $\Big(\mathfrak{V}^{\zeta}_+\oplus\mathfrak{V}^{\zeta}_-\Big)$}+1$ and $s_n=\mathcal{N}_{(-\infty,\tilde{c}_n)}\big(\big(\mbox{\footnotesize $H_{\epsilon}^{\nu_n}$}\big)^2\big)+1$, the right side - and thus the left side - of (\ref{preliminary_claim}) is bounded from below by $\upsilon$ and hence
\begin{align}
\mathcal{N}_{(-\infty,\upsilon)}\mbox{\footnotesize $\Big($}D^2-m^2\mathbb{I}\mbox{\footnotesize $\Big)$}\leq 2N+\mathcal{N}_{(-\infty,\upsilon/(1-\zeta))}\mbox{\footnotesize $\Big(\mathfrak{V}^{\zeta}_+\oplus\mathfrak{V}^{\zeta}_-\Big)$}+\displaystyle\sum\limits_{n=1}^N\mathcal{N}_{(-\infty,\tilde{c}_n)}\mbox{\footnotesize $\Big($}\big(\mbox{\footnotesize $H_{\epsilon}^{\nu_n}$}\big)^2\mbox{\footnotesize $\Big)$}
\end{align}
holds. As the spectra of $\left\{H^{\nu_n}_{\epsilon}\right\}_{n=1}^N$ are discrete (see Lem.~\ref{sa_ext_disc_sp}), $\mathcal{N}_{(-\infty,\tilde{c}_n)}\big(\big(\mbox{\footnotesize $H_{\epsilon}^{\nu_n}$}\big)^2\big)$ is finite for all $n\in\{1,\dots,N\}$. Then, the upper bound in (\ref{bounds_for_number_of_ev}) follows from 
\begin{align}
\mathcal{N}_{(-\infty,\upsilon/(1-\zeta))}\mbox{\footnotesize $\Big(\mathfrak{V}^{\zeta}_+\oplus\mathfrak{V}^{\zeta}_-\Big)$}\leq \mathcal{N}_{(-\infty,\upsilon)}\mbox{\footnotesize $\Big(\mathfrak{V}^{\zeta}_+\oplus\mathfrak{V}^{\zeta}_-\Big)$}= \displaystyle\sum\limits_{+,-}\mathcal{N}_{(-\infty,\upsilon)}\mbox{\footnotesize $\Big(\mathfrak{V}^{\zeta}_{\pm}\Big)$}\,.
\end{align}

As for the lower bound, by the Min-Max principle, the eigenvalues of $D^2-m^2\mathbb{I}$ are bounded from above by those of the Friedrichs extension of
\begin{align}
\textnormal{\textsf{C}}_0^{\infty}(\mathbb{R}^2\setminus \overline{B_{\gamma}},\mathbb{C}^2)\rightarrow \textnormal{\textsf{L}}^2(\mathbb{R}^2\setminus B_{\gamma},\mathbb{C}^2)\,,\quad \psi\mapsto (D^2-m^2\mathbb{I})\psi\,.
\end{align}
As in (\ref{cs_ineq_conseq}), we estimate for all $\psi\equiv (\psi_1,\psi_2)\in \textnormal{\textsf{C}}_0^{\infty}(\mathbb{R}^2\setminus\overline{B_{\gamma}},\mathbb{C}^2)$
\begin{align}
\begin{split}
\|D\psi\|^2-m^2\|\psi\|^2 & \leq \left\langle\psi,\left(-(1+\zeta)\Delta + 2mV\sigma_3+(1+1/\zeta)V^2\right)\psi\right\rangle\\
& = (1+\zeta)\left\langle \psi,\left(-\Delta+\left[2mV\sigma_3+(1+1/\zeta)V^2\right]/(1+\zeta)\right)\psi\right\rangle \\
& = (1+\zeta) \Big\langle \psi_1\oplus\psi_2,\left(\mathfrak{W}^{\zeta}_+\oplus\mathfrak{W}^{\zeta}_-\right)\psi_1\oplus\psi_2\Big\rangle 
\end{split}
\end{align}
(cf. \cite{Rademacher}, inequality (15)). Then, the lower bound in (\ref{bounds_for_number_of_ev}) follows from
\begin{align}
\mathcal{N}_{(-\infty,\upsilon/(1+\zeta))}\mbox{\footnotesize $\Big(\mathfrak{W}^{\zeta}_+\oplus\mathfrak{W}^{\zeta}_-\Big)$}\geq \mathcal{N}_{(-\infty,\upsilon)}\mbox{\footnotesize $\Big(\mathfrak{W}^{\zeta}_+\oplus\mathfrak{W}^{\zeta}_-\Big)$}= \displaystyle\sum\limits_{+,-}\mathcal{N}_{(-\infty,\upsilon)}\mbox{\footnotesize $\Big(\mathfrak{W}^{\zeta}_{\pm}\Big)$}\,.
\end{align}
\end{proof}

At the expense of a bounded and compactly supported localization error, the negative eigenvalues of Schr{\"o}dinger operators defined in $\textnormal{\textsf{L}}^2(\mathbb{R}^2)$ with pure long range dipole potentials can be bounded from below by those of Schr{\"o}dinger operators defined in $\textnormal{\textsf{L}}^2(\mathbb{R}^2\setminus\overline{B_{\gamma}})$ with pure dipole potentials (see below). The latter accumulate exponentially fast at the bottom of the essential spectrum (see~\cite{Rademacher}, Lem.~1). To decouple the interior from the exterior part, we make use of a further partition of unity $(\tilde{U}_{\textnormal{int}},\tilde{U}_{\textnormal{ext}})\in \textnormal{\textsf{C}}^{\infty}_0(B_{2\gamma},[0,1])\times\textnormal{\textsf{C}}^{\infty}(\mathbb{R}^2\setminus \overline{B_{\gamma}},[0,1])$ with $(\tilde{U}_{\textnormal{int}})^2+(\tilde{U}_{\textnormal{ext}})^2=1$.

\begin{Lem}\label{pure_dipole_lem}
Let $\mathfrak{c}\in\mathbb{R}^2\setminus\{0\}$ and $\mbox{\small $-\Delta_{\mathbb{R}^2\setminus\overline{B_{\gamma}}}^{\textnormal{D}}$}$ be the Dirichlet-Laplacian for $\mathbb{R}^2\setminus\overline{B_{\gamma}}$ (see~\cite{Reed4}, Section~XIII.15). Then it holds for all $\upsilon\in\mathbb{R}^-$ that
\begin{align}\label{pure_dipole_ineq}
\mathcal{N}_{(-\infty,\upsilon)}\left(-\Delta_{\mathbb{R}^2}+\chi_{\mathbb{R}^2\setminus B_{\gamma}}\langle\mathfrak{c},\cdot\rangle_{\mathbb{R}^2}/|\cdot|^3+\mathscr{L}^{\mathfrak{c}}_{\gamma}\right)\leq\mathcal{N}_{(-\infty,\upsilon)}\left(\mbox{\small $-\Delta_{\mathbb{R}^2\setminus\overline{B_{\gamma}}}^{\textnormal{D}}$}+\langle\mathfrak{c},\cdot\rangle_{\mathbb{R}^2}/|\cdot|^3\right)\,,
\end{align}
where $\mathscr{L}^{\mathfrak{c}}_{\gamma}:=\chi_{B_{2\gamma}\setminus B_{\gamma}} |\langle\mathfrak{c},\cdot\rangle_{\mathbb{R}^2}|/|\cdot|^3+|\nabla\tilde{U}_{\textnormal{int}}|^2+|\nabla\tilde{U}_{\textnormal{ext}}|^2$ is the localization error.
\end{Lem}
\begin{proof}
Let $M\subset\textnormal{\textsf{L}}^2(\mathbb{R}^2)$. As in (\ref{loc_est_3}), we estimate using the IMS localization formula for Schr{\"o}dinger operators (see \cite{Cycon}, Thm. 3.2)
\begin{align}
\begin{split}
&\inf\limits_{\psi\in \textnormal{\scriptsize\textsf{C}\small}_0^{\infty}(\mathbb{R}^2)\cap (\tilde{U}_{\textnormal{ext}}M)^{\perp}\atop \|\psi\|=1} \left\langle\psi,\left(-\Delta+\chi_{\mathbb{R}^2\setminus B_{\gamma}}\langle\mathfrak{c},\cdot\rangle_{\mathbb{R}^2}/|\cdot|^3+\mathscr{L}^{\mathfrak{c}}_{\gamma}\right)\psi\right\rangle\\
&=  \inf\limits_{\psi\in \textnormal{\scriptsize\textsf{C}\normalsize}_0^{\infty}(\mathbb{R}^2)\cap (\tilde{U}_{\textnormal{ext}}M)^{\perp}\atop \|\psi\|=1}\mbox{\small $\Big[\left\langle \tilde{U}_{\textnormal{ext}}\psi,\left(-\Delta+\chi_{\mathbb{R}^2\setminus B_{\gamma}}\langle\mathfrak{c},\cdot\rangle_{\mathbb{R}^2}/|\cdot|^3+\chi_{B_{2\gamma}\setminus B_{\gamma}} |\langle\mathfrak{c},\cdot\rangle_{\mathbb{R}^2}|/|\cdot|^3\right)\tilde{U}_{\textnormal{ext}}\psi\right\rangle+$}\\
&\qquad\quad\quad\quad\quad\quad + \mbox{\small $\left\langle \tilde{U}_{\textnormal{int}}\psi,\left(-\Delta+\chi_{\mathbb{R}^2\setminus B_{\gamma}}\langle\mathfrak{c},\cdot\rangle_{\mathbb{R}^2}/|\cdot|^3+\chi_{B_{2\gamma}\setminus B_{\gamma}} |\langle\mathfrak{c},\cdot\rangle_{\mathbb{R}^2}|/|\cdot|^3\right)\tilde{U}_{\textnormal{int}}\psi\right\rangle\Big]$}\\\\
&\geq \inf\limits_{\psi\in \textnormal{\scriptsize\textsf{C}\normalsize}_0^{\infty}(\mathbb{R}^2)\cap (\tilde{U}_{\textnormal{ext}}M)^{\perp}\atop \|\psi\|=1}\left\langle \tilde{U}_{\textnormal{ext}}\psi,\left(-\Delta+\chi_{\mathbb{R}^2\setminus B_{\gamma}}\langle\mathfrak{c},\cdot\rangle_{\mathbb{R}^2}/|\cdot|^3\right)\tilde{U}_{\textnormal{ext}}\psi\right\rangle\\
&\geq  \inf\limits_{\psi\in \textnormal{\scriptsize\textsf{C}\normalsize}_0^{\infty}(\mathbb{R}^2\setminus \overline{B_{\gamma}})\cap M^{\perp}\atop \|\psi\|\leq 1}\left\langle\psi, \left(-\Delta+\langle\mathfrak{c},\cdot\rangle_{\mathbb{R}^2}/|\cdot|^3\right)\psi\right\rangle\,.
\end{split}
\end{align}
By an estimate similar to (\ref{loc_est_2}), we conclude that negative eigenvalues satisfy
\begin{align}
\mu_{s}\left(-\Delta_{\mathbb{R}^2}+\chi_{\mathbb{R}^2\setminus B_{\gamma}}\langle\mathfrak{c},\cdot\rangle_{\mathbb{R}^2}/|\cdot|^3+\mathscr{L}^{\mathfrak{c}}_{\gamma}\right)\geq \mu_{s}\left(\mbox{\small $-\Delta_{\mathbb{R}^2\setminus\overline{B_{\gamma}}}^{\textnormal{D}}$}+\langle\mathfrak{c},\cdot\rangle_{\mathbb{R}^2}/|\cdot|^3\right)\,,
\end{align}
which implies (\ref{pure_dipole_ineq}).
\end{proof}

Next, following Rademacher and Siedentop~\cite{Rademacher}, we decouple the pure dipole part from higher-order multipole moments, which - a posteriori - merely contribute with finitely many negative eigenvalues. For this purpose, we formulate the follo\-wing statement.

\begin{Lem}\label{kirsch_lem}
Suppose, $A_1$, $A_2$ and $A_3$ are lower semi-bounded self-adjoint operators in a Hilbert space with a common form core $K$ such that $\textnormal{inf}\sigma_{\textnormal{ess}}(A_j)\in\mathbb{R}^+_0$ for $j=1,2,3$ is satisfied and $A_1=A_2+A_3$ holds in the form sense on $K$, i.e., 
\begin{align}
\left\langle \psi,A_1\psi\right\rangle=\left\langle \psi,A_2\psi\right\rangle+\left\langle \psi,A_3\psi\right\rangle\quad\forall \psi\in K\,.
\end{align}
Then, it holds for all $\eta\in (0,1)$ and $\upsilon\in\mathbb{R}^-$ that
\begin{align}\label{kirsch_statement}
\mathcal{N}_{(-\infty,\upsilon)}\left(A_1\right)\leq \mathcal{N}_{(-\infty,(1-\eta)\upsilon)}\left(A_2\right)+\mathcal{N}_{(-\infty,\eta\upsilon)}\left(A_3\right)\,.
\end{align}
\end{Lem}
\begin{proof}
The statement follows by mimicking the proof of Prop. 4. in~\cite{Kirsch}.
\end{proof}

\begin{Rem}
We obtain for all $\xi\in (0,1)$ and $\upsilon^{\prime}\in\mathbb{R}^-$ the inequality
\begin{align}\label{kirsch_variant}
\mathcal{N}_{(-\infty,\upsilon^{\prime})}\left(A_2\right)\geq \mathcal{N}_{(-\infty,(1+\xi)\upsilon^{\prime})}\left(A_1\right)-\mathcal{N}_{(-\infty,\xi\upsilon^{\prime})}\left(A_3\right)
\end{align}
when we insert $\upsilon=(1+\xi)\upsilon^{\prime}$ and $\eta=(1+\xi)^{-1}\xi$ into (\ref{kirsch_statement}).
\end{Rem}

Let $\eta,\xi,\zeta\in (0,1)$. We decompose $\mathfrak{V}^{\zeta}_{\pm}$ into $\mathfrak{V}^{\zeta}_{\pm}=(1-\eta)\mathfrak{X}^{\zeta,\eta}_{\pm}+\eta\mathfrak{T}^{\zeta,\eta}_{\pm}$, where
\begin{align}
\mathfrak{X}^{\zeta,\eta}_{\pm}:=-\Delta_{\mathbb{R}^2}\pm \chi_{\mathbb{R}^2\setminus B_{\gamma}}\frac{2m}{(1-\zeta)(1-\eta)}\frac{\langle\mathfrak{d},\cdot\rangle_{\mathbb{R}^2}}{|\cdot|^3}+\mathscr{L}^{\pm 2m\mathfrak{d}/[(1-\zeta)(1-\eta)]}_{\gamma}
\end{align}
and
\begin{align}
\begin{split}
\mathfrak{T}^{\zeta,\eta}_{\pm}:=-\Delta_{\mathbb{R}^2}+\bigg(\pm 2m\left[V\chi_{\textnormal{supp}(U_0)}-\frac{\langle\mathfrak{d},\cdot\rangle_{\mathbb{R}^2}}{|\cdot|^3}\chi_{\mathbb{R}^2\setminus B_{\gamma}}\right]+(1-1/\zeta)V^2\chi_{\textnormal{supp}(U_0)}-\\
-\displaystyle\sum\limits_{n=0}^N|\nabla U_n|^2-(1-\eta)(1-\zeta)\mathscr{L}_{\gamma}^{\pm 2m\mathfrak{d}/[(1-\eta)(1-\zeta)]}\bigg)/[(1-\zeta)\eta]\,.
\end{split}
\end{align}
Since $V\chi_{\textnormal{supp}(U_0)}$ and $V^2\chi_{\textnormal{supp}}(U_0)$ are $(-\Delta_{\mathbb{R}^2})$-compact (see above), $\mathfrak{T}^{\zeta,\eta}_{\pm}$ is self-adjoint and bounded from below.\\
Using Lem.~\ref{kirsch_lem} and then Lem.~\ref{pure_dipole_lem}, we obtain for all $\upsilon\in\mathbb{R}^-$ that
\begin{align}\label{decouple_above}
\begin{split}
\mathcal{N}_{(-\infty,\upsilon)}\mbox{\small $\Big(\mathfrak{V}^{\zeta}_{\pm}\Big) $}& \leq \mathcal{N}_{(-\infty,\upsilon)}\mbox{\small $\Big(\mathfrak{X}^{\zeta,\eta}_{\pm}\Big) $}+\mathcal{N}_{(-\infty,\upsilon)}\mbox{\small $\Big(\mathfrak{T}^{\zeta,\eta}_{\pm}\Big) $}\\
& \leq \mathcal{N}_{(-\infty,\upsilon)}\mbox{\small $\Big(-\Delta_{\mathbb{R}^2\setminus \overline{B_{\gamma}}}^{\textnormal{D}}\pm\frac{2m}{(1-\zeta)(1-\eta)}\frac{\langle\mathfrak{d},\cdot\rangle_{\mathbb{R}^2}}{|\cdot|^3}\Big) $}+\mathcal{N}_{\mathbb{R}^-}\mbox{\small $\Big(\mathfrak{T}^{\zeta,\eta}_{\pm}\Big) $}\,.
\end{split}
\end{align}
We decompose $\mathfrak{W}^{\zeta}_{\pm}$ in a similar way. Let $\mathfrak{Z}^{\zeta,\xi}_{\pm}$ be the Friedrichs extension of
\begin{align}\label{rest_exterior}
\mbox{\footnotesize $\textnormal{\textsf{C}}^{\infty}_0(\mathbb{R}^2\setminus \overline{B_{\gamma}})\rightarrow\textnormal{\textsf{L}}^2(\mathbb{R}^2\setminus B_{\gamma})\,, \quad \psi\mapsto-\Delta\psi+\left[\mp 2m\left[V-\mbox{\small $\frac{\langle\mathfrak{d},\cdot\rangle_{\mathbb{R}^2}}{|\cdot|^3}$}\right]-(1+1/\zeta)V^2\right]\psi/[(1+\zeta)\xi]\,,$}
\end{align}
which is bounded from below, since $\chi_{\mathbb{R}^2\setminus B_{\gamma}}V$ and $\chi_{\mathbb{R}^2\setminus B_{\gamma}}V^2$ are $(-\Delta_{\mathbb{R}^2})$-compact (cf. above considerations).\\
Then, since 
\begin{align}
(1+\xi)\left(-\Delta^{\textnormal{D}}_{\mathbb{R}^2\setminus \overline{B_{\gamma}}}\pm \mbox{\small $\frac{2m}{(1+\zeta)(1+\xi)}\frac{\langle\mathfrak{d},\cdot\rangle_{\mathbb{R}^2}}{|\cdot|^3}$}\right)=\mathfrak{W}^{\zeta}_{\pm}+\xi\mathfrak{Z}^{\zeta,\xi}_{\pm}
\end{align}
holds in the form sense on the common form core $\textnormal{\textsf{C}}^{\infty}_0(\mathbb{R}^2\setminus \overline{B_{\gamma}})$, we obtain
\begin{align}\label{decouple_below}
\begin{split}
\mathcal{N}_{(-\infty,\upsilon)}\mbox{\small $\Big(\mathfrak{W}^{\zeta}_{\pm}\Big) $}& \geq \mathcal{N}_{(-\infty,\upsilon)}\mbox{\small $\Big(-\Delta_{\mathbb{R}^2\setminus \overline{B_{\gamma}}}^{\textnormal{D}}\pm\frac{2m}{(1+\zeta)(1+\xi)}\frac{\langle\mathfrak{d},\cdot\rangle_{\mathbb{R}^2}}{|\cdot|^3}\Big) $}-\mathcal{N}_{(-\infty,\upsilon)}\mbox{\small $\Big(\mathfrak{Z}^{\zeta,\xi}_{\pm}\Big) $}\\
& \geq \mathcal{N}_{(-\infty,\upsilon)}\mbox{\small $\Big(-\Delta_{\mathbb{R}^2\setminus \overline{B_{\gamma}}}^{\textnormal{D}}\pm\frac{2m}{(1+\zeta)(1+\eta)}\frac{\langle\mathfrak{d},\cdot\rangle_{\mathbb{R}^2}}{|\cdot|^3}\Big) $}-\mathcal{N}_{\mathbb{R}^-}\mbox{\small $\Big(\mathfrak{Z}^{\zeta,\xi}_{\pm}\Big) $}\,.
\end{split}
\end{align}
for all $\upsilon\in\mathbb{R}^-$ by using (\ref{kirsch_variant}).\\

As mentioned above, we now show that the higher-order multipole moments merely contribute with finitely many negative eigenvalues.

\begin{Lem}\label{finite_number_ev}
Let $\zeta,\eta,\xi\in (0,1)$. Then, $\mathcal{N}_{\mathbb{R}^-}\mbox{\small $\Big(\mathfrak{T}^{\zeta,\eta}_{\pm}\Big)$}$ and $\mathcal{N}_{\mathbb{R}^-}\mbox{\small $\Big(\mathfrak{Z}^{\zeta,\xi}_{\pm}\Big)$}$ are finite.
\end{Lem}
\begin{proof}
It follows from Hypothesis \ref{hyp5}.(a) in Thm. \ref{main_theorem} that $V\chi_{\textnormal{supp}(U_0)}-\mbox{\small $\frac{\langle\mathfrak{d},\cdot\rangle_{\mathbb{R}^2}}{|\cdot|^3}$}\chi_{\mathbb{R}^2\setminus B_{\gamma}}$ and $V^2\chi_{\textnormal{supp}(U_0)}$ - and thus also the potential of $\mathfrak{T}^{\zeta,\eta}_{\pm}$ - lie in $\textnormal{\textsf{L}}^1(\mathbb{R}^2;\log(2+|x|)\textnormal{d}x)$. Accordingly, Hypothesis \ref{hyp5}.(b) in Thm. \ref{main_theorem} implies that their spherical rearrangements are contained in $\textnormal{\textsf{L}}^1(\mathbb{R}^+;\log_+(r^{-1})\textnormal{d}r)$. Then, the finiteness of $\mathcal{N}_{\mathbb{R}^-}\mbox{\small $\Big(\mathfrak{T}^{\zeta,\eta}_{\pm}\Big)$}$ follows from Thm. 4.3 in~\cite{Shargorodsky}. The same applies to the zero extension of the potential in~(\ref{rest_exterior}) to $\mathbb{R}^2$. Hence - by the inclusion of form cores $\textnormal{\textsf{C}}^{\infty}_0(\mathbb{R}^2\setminus\overline{B_{\gamma}})\subset \textnormal{\textsf{C}}^{\infty}_0(\mathbb{R}^2)$ - Thm.~4.3 in~\cite{Shargorodsky} also implies that $\mathcal{N}_{\mathbb{R}^-}\mbox{\small $\Big(\mathfrak{Z}^{\zeta,\xi}_{\pm}\Big)$}$ is finite.
\end{proof}

We are now prepared for the proof of Thm. \ref{main_theorem}:

\begin{proof}
Let $\zeta,\eta,\xi\in (0,1)$. Using Lem. \ref{lem_bounds_number_of_ev} and \ref{finite_number_ev} and inequalities (\ref{decouple_above}) and (\ref{decouple_below}) , we estimate
\begin{align}
\begin{split}
& \limsup\limits_{E\nearrow m}\frac{\mathcal{N}_{(-\infty,E^2-m^2)}\mbox{\footnotesize $\Big($}D^2-m^2\mathbb{I}\mbox{\footnotesize $\Big)$}}{|\log(m-E)|}\leq\\
\leq & \displaystyle\sum\limits_{+,-}\limsup\limits_{E\nearrow m}\frac{\mathcal{N}_{(-\infty,E^2-m^2)}\Big(-\Delta_{\mathbb{R}^2\setminus\overline{B_{\gamma}}}^{\textnormal{D}}\pm \frac{2m}{(1-\zeta)(1-\eta)}\frac{\langle\mathfrak{d},\cdot\rangle_{\mathbb{R}^2}}{|\cdot|^3}\Big)}{|\log(m^2-E^2)|}\overbrace{\left|\frac{\log(m^2-E^2)}{\log(m-E)}\right|}^{\rightarrow 1}
\end{split}
\end{align}
and 
\begin{align}
\begin{split}
& \liminf\limits_{E\nearrow m}\frac{\mathcal{N}_{(-\infty,E^2-m^2)}\mbox{\footnotesize $\Big($}D^2-m^2\mathbb{I}\mbox{\footnotesize $\Big)$}}{|\log(m-E)|}\geq\\
\geq & \displaystyle\sum\limits_{+,-}\liminf\limits_{E\nearrow m}\frac{\mathcal{N}_{(-\infty,E^2-m^2)}\Big(-\Delta_{\mathbb{R}^2\setminus\overline{B_{\gamma}}}^{\textnormal{D}}\pm \frac{2m}{(1+\zeta)(1+\xi)}\frac{\langle\mathfrak{d},\cdot\rangle_{\mathbb{R}^2}}{|\cdot|^3}\Big)}{|\log(m^2-E^2)|}\overbrace{\left|\frac{\log(m^2-E^2)}{\log(m-E)}\right|}^{\rightarrow 1}\,.
\end{split}
\end{align}
Due to the continuity of $\mbox{\small $\textnormal{tr}\Big(\mbox{\small $\sqrt{(M_{(\cdot)})_-}$}\hspace{0.5mm}\Big)$}$ (see \cite{Rademacher}), the desired result follows from Lem. 1 in \cite{Rademacher} in the limits $\zeta,\eta\rightarrow 0$ and $\zeta,\xi\rightarrow 0$, respectively.
\end{proof}

\section*{Acknowledgements}
\noindent I thank Heinz Siedentop and Sergey Morozov for many useful discussions.

\section*{References}
\newcommand\oldsection{}
\let\oldsection=\section
\renewcommand{\section}[2]{}

\let\section=\oldsection

\end{document}